\documentclass[journal,draftcls,onecolumn,12pt,twoside]{IEEEtranTCOM}
\IEEEoverridecommandlockouts

\normalsize
\usepackage{mathtools}

\usepackage{ifpdf}
\usepackage{amsthm}
\usepackage{amsmath}
\usepackage{nicefrac}
\usepackage{cite}
\usepackage{amsfonts}
\usepackage{comment}
\usepackage{subcaption}
 \usepackage[font=small,labelfont=bf]{caption}
\usepackage{url}
\usepackage{amssymb}
\usepackage[paper=letterpaper,margin=0.75in]{geometry}
\usepackage[ansinew]{inputenc}

\usepackage{lipsum}
\usepackage{hyperref}

\theoremstyle{definition}
\newtheorem{thm}{Theorem}
\newtheorem{definition}{Definition}
\newtheorem{lem}{Lemma}

\newcommand*{\QEDA}{\hfill\ensuremath{\blacksquare}}%

\makeatletter
\renewenvironment{proof}[1][\proofname] {\par\pushQED{\qed}\normalfont\topsep6\p@\@plus6\p@\relax\trivlist\item[\hskip\labelsep\bfseries#1\@addpunct{.}]\ignorespaces}{\popQED\endtrivlist\@endpefalse}
\makeatother

\begin{document}
\title{Fundamental Limits of Covert Packet Insertion}

\author{
	\IEEEauthorblockN{Ramin Soltani\IEEEauthorrefmark{1},
		Dennis Goeckel\IEEEauthorrefmark{1}, Don Towsley\IEEEauthorrefmark{2}, and Amir Houmansadr\IEEEauthorrefmark{2}}
	
	\IEEEauthorblockA{\IEEEauthorrefmark{1}Electrical~and~Computer~Engineering~Department,~University~of~Massachusetts,~Amherst,
		\{soltani, goeckel\}@ecs.umass.edu\\}
	\IEEEauthorblockA{\IEEEauthorrefmark{2}College of Information and Computer Sciences, University of Massachusetts, Amherst,
		\{towsley, amir\}@cs.umass.edu}
	
	\thanks{ This work has been supported by the National Science Foundation under grants ECCS-1309573 and CNS-1525642, and appeared, in part, at the Allerton Conference on Communications, Control, and Computing in 2015 \cite{soltani2015covert} and 2016 \cite{soltani2016allerton}.}	
	 \thanks{This work has been submitted to the IEEE for possible publication. Copyright may be transferred without notice, after which this version may no longer be accessible.}
}

\markboth{IEEE Transactions on Communications}%
{Submitted paper}
\maketitle
\begin{abstract}
Covert communication conceals the existence of the transmission from a watchful adversary. We consider the fundamental limits for covert communications via packet insertion over packet channels whose packet timings are governed by a renewal process of rate $\lambda$. Authorized transmitter Jack sends packets to authorized receiver Steve, and covert transmitter Alice wishes to transmit packets to covert receiver Bob without being detected by watchful adversary Willie. Willie cannot authenticate the source of the packets. Hence, he looks for statistical anomalies in the packet stream from Jack to Steve to attempt detection of unauthorized packet insertion. First, we consider a special case where the packet timings are governed by a Poisson process and we show that Alice can covertly insert $\mathcal{O}(\sqrt{\lambda T})$ packets for Bob in a time interval of length $T$; conversely, if Alice inserts $\omega(\sqrt{\lambda T})$, she will be detected by Willie with high probability. Then, we extend our results to general renewal channels and show that in a stream of $N$ packets transmitted by Jack, Alice can covertly insert $\mathcal{O}(\sqrt{N})$ packets; if she inserts $\omega(\sqrt{N})$ packets, she will be detected by Willie with high probability. 
\end{abstract}

\IEEEpeerreviewmaketitle

\section{Introduction}
\IEEEPARstart{P}{rivacy} and security have become crucial issues in daily life as the use of communication systems has increased (e.g. telephone, email, social media) ~\cite{lopez2008wireless,nazanin_ISIT2017
	,hadian2016privacy,takbiri2018asymptotic
	hadian2018privacy,nichols2001wireless
}. Information theoretic secrecy~\cite{bloch11pls} and encryption~\cite{talb2006} protect the secrecy of message contents; however, these techniques do not satisfy the security and privacy requirements of users in many scenarios. Recently, the need for another level of secrecy was highlighted by the Snowden disclosures \cite{snowden}: users of a communication system often need not only secrecy for the contents of their messages, but also for hiding the existence of their communication. As a solution, covert communication ensures that a watchful adversary is not able to detect whether communication is taking place or not. Two applications of covert communication are the removal of the ability to track daily user activities and to hide the presence of military activities.

Steganography~\cite{ker07pool} is utilized to covertly embed information into an overt message on a digital (and typically noiseless) channels. Alternatively, spread spectrum methods~\cite{simon1994spread} provide covert communication on noisy channels. Information-theoretic limits of covert communications only recently gained attention first with the study of additive white Gaussian (AWGN) channels \cite{bash_isit2012,bash_jsac2013}, which was later extended to provide a comprehensive characterization of the limits of covert communication over discrete memoryless channels (DMCs), optical channels, and AWGN channels \cite{bash2015hiding,goeckel2016covert,bash_isit2013,  soltani2014covert,soltani2018covert, jaggi_isit2014,bloch2016covert,wang15covert-isit,lee14posratecovert, jaggi_uncertain,arumugam2018covert,8278022,tahmasbi2018covert}.

In this paper, we extend the work in~\cite{bash2015hiding,goeckel2016covert,bash_isit2013,  soltani2014covert,soltani2018covert, jaggi_isit2014,bloch2016covert,wang15covert-isit,lee14posratecovert, jaggi_uncertain,arumugam2018covert,8278022,tahmasbi2018covert} to packet processes typical of wired computer networks. In computer networks, covert channels can be divided into two major categories ~\cite{NIST:NCSC}: \emph{covert storage channels} and \emph{covert timing channels}. A covert storage channel involves the writing of a shared storage location by one process and reading of it by another; e.g. modifying headers of packets \cite{handel1996hiding,rowland1997covert,giffin2002covert,murdoch2005embedding}. Alternatively, a covert timing channel involves the exchange of information between two users by manipulation of timings of some shared resources; e.g. embedding information packet timings first explored by Girling~\cite{Girling87} and later studied by many others ~\cite{Cabuk:2004:ICT:1030083.1030108,berk2005detection,shah2006keyboards,cabuk2006network,liu2009hide,liu2010robust,houmansadr2011coco,Arch2012,liu2012network}. This includes applications of covert channels in TCP/IP!\cite{murdoch2005embedding,sekke2009,mazurczyk2012lost}, VoIP~\cite{mazurczyk2008covert}, LTE-A~\cite{rezaei2013analysis},
BitTorrent~\cite{cunche2014asynchronous}, and establishment of a covert communication over IPV4~\cite{Cabuk:2004:ICT:1030083.1030108,ahsan2002practical} and IPV6~\cite{lucena2005covert} have been studied.

Considerable work has focused on detection of covert  channels~\cite{berk2005detection, gianvecchio2007detecting,rezaei2017towards,rezaei2015novel,helouet2011covert,gianvecchio2007detecting,Zhao2015,kemmerer1983shared} as well as eluding detection by leveraging the statistical properties of the legitimate channel~\cite{gianvecchio2008model}. Moreover, significant research has been performed on quantifying and optimizing the capacity of covert channels~\cite{millen1987covert, millen1989finite,Anand1998inf,berk2005detection,martin2006noisy,sekke2007capacity,sekke2009,Mosko92Capac,Moskowitz1994} by leveraging information-theoretic analysis and the use of various coding techniques \cite{kiyavash2009J,kiyavash2009A,Arch2012}. In particular, Anantharam and Verdu~\cite{verdubitsq} derived the Shannon capacity of the timing channel with a single-server queue, and Dunn~\cite{dunn2009} analyzed the secrecy capacity of such a system. 

Per above, here we take a fundamental approach analogous to~\cite{bash_isit2012,bash_jsac2013,bash2015hiding,goeckel2016covert,bash_isit2013, soltani2014covert,soltani2018covert, jaggi_isit2014,bloch2016covert,wang15covert-isit,lee14posratecovert, jaggi_uncertain}, but turn our attention to covert communication over wired channels in which communication takes place through packet transmissions. Specifically, we consider the scenario shown in Fig.~\ref{fig:SysMod2}. Authorized transmitter Jack sends packets to authorized receiver Steve. Assume that Alice wishes to transmit data covertly to Bob on this channel in the presence of an adversary, Willie, who is monitoring the channel. Willie can be in one of the two locations, either between Alice and Bob (Setting 1), or between Bob and Steve (Setting 2). Alice and Bob know that Willie is located at one or the other of these two places; however, they do not know which place he is located at. We assume Willie cannot authenticate the source of the packets (e.g., whether they are sent by Jack or not). However, he knows the statistical model of the timings of the packets transmitted by Jack. Alice can buffer and release Jack's packets and insert her own packets. Also, Bob can authenticate packets, remove the ones originally inserted by Alice, and buffer and release Jack's packets. We assume Alice can only send information to Bob by inserting her own packets into the channel, since she is not allowed to share a secret codebook with Bob and thus she is not able to send covert messages to Bob via packet timings; i.e., altering the timing of the packets according to a shared codebook, to embed information in inter-packet delays (IPDs)~\cite{verdubitsq}. In addition, transmission of information through packet timings is sensitive to natural network noise and thus is not applicable in scenarios where timing noise alters the transmitted packet timings (codeword) such that the receiver is not able to decode the message with the required decoding error (e.g. complex channels where the packet streams are first mixed and then separated). We answer this question: how many packets can Alice transmit to Bob without being detected by Willie? 

We consider two statistical models for the timing process of Jack's transmitted packets. First, we analyze a Poisson channel (Assumption 1)~\cite{soltani2015covert}; i.e., IPDs of Jack's transmitted stream are modeled by independent and identically distributed (i.i.d.) exponential random variables with mean $\lambda^{-1}$, and Willie is aware of this. Therefore, Willie seeks to verify whether the packet process has the proper characteristics. We exploit the fact that the superposition of two independent Poisson processes is a Poisson process: Alice generates a Poisson process of low enough rate and uses it to govern the times at which she inserts the covert packets into the Jack-to-Steve channel. We assume Willie is aware of Alice's transmission strategy (insertion scheme, rate, etc.) as well as what Bob can do, if they choose to communicate with each other.

Covertness as defined formally in Section II requires that Willie's decision on whether Alice transmits or not be arbitrarily close to random guessing. In Theorem~\ref{poissonins}, we show that Alice can transmit $\mathcal{O} \left(\sqrt{\lambda T}\right)$ packets covertly to Bob in a time interval of length $T$. Conversely, we prove that if Alice transmits ${\omega}\left(\sqrt{\lambda T}\right)$ packets during a time interval of length $T$, she will be detected by Willie with high probability. 
\begin{figure} 

	\centering
	\begin{subfigure}[b]{\textwidth}
		\includegraphics[width=\textwidth,height=\textheight,keepaspectratio]{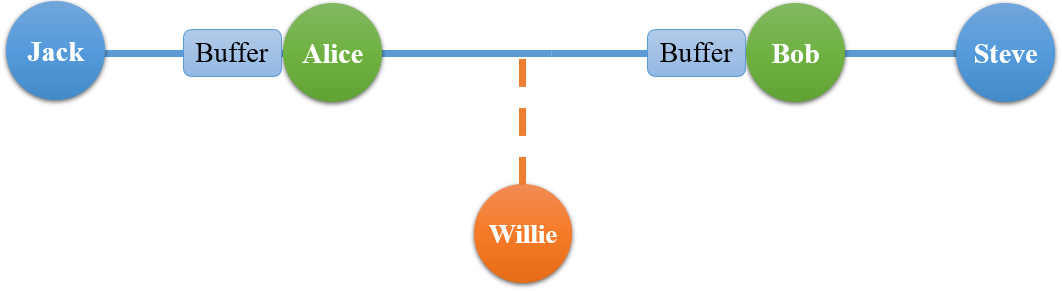}
		\caption{Setting 1}
		\label{fig:SysMod2a} 
	\end{subfigure}
	
	\begin{subfigure}[b]{\textwidth}
		\includegraphics[width=1\linewidth]{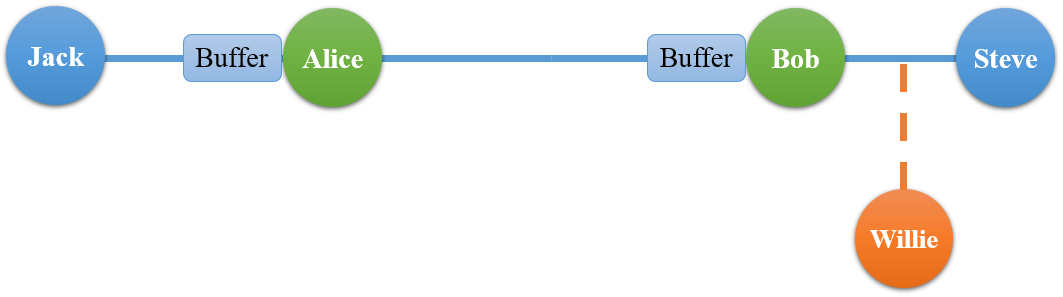}
		\caption{Setting 2}
		\label{fig:SysMod2b}
	\end{subfigure}
	\caption[.]{System configuration: Authorized transmitter Jack sends packets to authorized receiver Steve, and covert transmitter Alice uses the same channel to send her own packets to Bob without being detected by warden Willie, who cannot authenticate the source of the packets. Willie is aware of the statistical model of the timing of the packets sent by Jack. Willie can be in one of the two locations: (a) Setting 1: between Alice and Bob; or (b) Setting 2: between Bob and Steve. Willie's location (i.e. the setting) is not known to Alice or Bob.}\label{fig:SysMod2}
\end{figure}
Next, we extend the Poisson channel to a renewal channel~\cite{soltani2016allerton} (Assumption 2), where the timings of Jack's transmitted packets are modeled by a renewal process; i.e., IPDs of Jack's transmitted stream are modeled by i.i.d. random variables with probability density function (pdf) $p(x)$ and transmission rate $\lambda=\left(\int_{x=0}^{\infty}x p(x)dx\right)^{-1}$ packets per second, and Willie is aware of these characteristics. Therefore, Willie seeks to verify whether the packet process has the proper properties. Since the superposition of two independent renewal processes is a not generally a renewal process, we use a technique different from the one employed in the Poisson channel.

The remainder of the paper is organized as follows. In Section~\ref{modcons}, we present the system model and definitions employed. We provide constructions and their analysis for the Poisson channel in Section~\ref{sec:p1}, and we analyze the renewal channel in Section~\ref{sec:r1}. Section~\ref{sec:dis1} contains the discussion of the results, and Section~\ref{sec:con1} summarizes our results. 
\section{System Model and Definitions}
\label{modcons}

\subsection{System Model} 
As shown in Fig.~\ref{fig:SysMod2}, Jack transmits packets to Steve while a watchful warden Willie observes the packet flow from his vantage point. Willie does not have access to the contents of the packets and therefore cannot authenticate whether a packet is originally transmitted by Jack, or generated and inserted by Alice. Instead, based on the timings of the packets, Willie attempts to discern any irregularities that might indicate that someone is inserting packets into the channel. Alice's goal is to insert her own packets in the stream of the packets sent by Jack so as to communicate covertly with Bob. Willie's location is fixed; he is either between Alice and Bob (Setting 1 shown in Fig.~\ref{fig:SysMod2a}), or he is between Bob and Steve (Setting 2 shown in Fig.~\ref{fig:SysMod2b}), and Alice and Bob are unaware of his location. 

Alice communicates with Bob by sending her packets into the channel, but Alice and Bob do not share a secret, thus preventing the distribution of a secret codebook to communicate via packet timings~\cite{verdubitsq,soltani2015covert,soltani2016allerton}. Alice can also buffer and release Jack's transmitted packets. Bob can authenticate, receive and remove packets originally inserted by another party. He is also allowed to buffer and release Jack's transmitted packets.  We assume Willie knows the characteristics of Alice's potential insertion scheme (rate, method of insertion, etc.) and Bob's capabilities. We denote the IPDs of the packets departing Jack, Alice, and Bob by $\{A_1^{(J)}, A_2^{(J)},\ldots\}$, $\{A_1^{(A)}, A_2^{(A)},\ldots\}$, and $\{A_1^{(B)}, A_2^{(B)},\ldots\}$, respectively.

We consider two sets of assumptions regarding the timing process of Jack's packets:

{\em Assumption 1 (Poisson channel model)} Transmission times for the packets generated by Jack are modeled by a Poisson process with parameter $\lambda$; i.e., IPDs of Jack's transmitted stream are i.i.d. random variables with pdf $p(x)=\lambda e^{-\lambda x}$, and Jack's packet transmission rate $\lambda$ is known to both Alice and Willie. 

{\em Assumption 2 (Renewal channel model):} Transmission times of the packets transmitted by Jack are modeled by a renewal process; i.e., IPDs of Jack's transmitted stream are positive i.i.d. random variables with pdf $p(x)$ and Jack's transmission rate is $\lambda=\left(\int_{0}^{\infty}x p(x)dx\right)^{-1}$. Both Willie and Alice know $p(x)$ and $\lambda$. 

When IPDs $A_1, A_2, \ldots$ are samples of $f(x)$ and modeled by a renewal process, the arrival times are $\tau_{f}(1), \tau_{f}(2), \ldots$, where
\begin{align}
\label{def6e1} \tau_{f}(i)=\sum_{j=1}^{i}A_j, 
\end{align}
\noindent and the total number of arrivals within the interval $[0,t]$ is $X_{f}(t) = \sup \left\{i: \tau_{f}(i) \leq t\right\}$. Observe:
\begin{align}
\label{def6e3} \{\tau_{f}(i) \leq t\} = \{X_{f}(t) \geq i\}.
\end{align}
\noindent For a Poisson process, ($f(x)=\lambda e^{-\lambda x}$), we omit the subscripts of $\tau_{f}(i)$ and $X_{f}(t)$.

\subsection{Definitions}
Willie is faced with a binary hypothesis test: the null hypothesis $H_0$ corresponds to the case that Alice does not transmit, and the alternative hypothesis $H_1$ corresponds to the case that Alice transmits. We denote the distributions of IPDs that Willie observes by $\mathbb{P}_{1}$ and $\mathbb{P}_{0}$ under $H_1$ and $H_0$, respectively. 

We denote by $\mathbb{P}_{FA}$ the probability of rejecting $H_0$ when it is true (type I error or false alarm), and $\mathbb{P}_{MD}$ the probability of rejecting $H_1$ when it is true (type II error or missed detection). 
Willie uses classical hypothesis testing and seeks to minimize $\mathbb{P}_{FA} + \mathbb{P}_{MD}$. 

Similar to the definition of covertness in~\cite{bash_isit2013, soltani2014covert,soltani2015covert,soltani2016allerton,soltani2018covert,soltani2018allerton,soltani2019dissertation}, and invisibility in~\cite{soltani2017towards,soltani2018fundamental}, we define covertness:

\begin{definition} (Covertness) Alice and Bob's communication is {\em covert}, if and only if Willie's sum of probabilities of error $\mathbb{P}_{FA}+\mathbb{P}_{MD}$ is lower bounded by $1-\epsilon$ for any $\epsilon>0$~\cite{bash_jsac2013}, for Willie in each of his possible locations.
\end{definition}

We present results under the assumption that $\mathbb{P}(H_0)=\mathbb{P}(H_1)=1/2$. However, this results in covertness for the general case~\cite[Appendix A]{soltani2018fundamentalbit}.



We use standard ``Big O'',``Little Omega'', and ``Big Theta'' notations~\cite{cormen2009introduction}.

\section{Poisson Channels (Assumption 1)}\label{sec:p1}
In this section, we consider the fundamental limits of covert packet insertion for the Poisson channel (Assumption 1). As evident from the proof, the possibility that Willie is located after Bob (Setting 2) is trivially addressed under Assumption 1. We will see this is not the case for the renewal channel model considered in Section~\ref{sec:r1}.
\begin{thm}
	\label{poissonins}
	In a Poisson channel with rate $\lambda$, Alice can covertly insert insert $\mathcal{O}(\sqrt{\lambda T})$ packets in a time interval of length $T$. Conversely, if Alice attempts to insert $\omega\left(\sqrt{\lambda T}\right)$ packets in a time interval of length $T$, there exists a detector that Willie can use to detect her with arbitrarily low sum of error probabilities $\mathbb{P}_{FA} + \mathbb{P}_{MD}$.
\end{thm}

\begin{proof} {\it (Achievability)} 
	
	\textit{\textbf{Construction}}: Alice generates a Poisson process with parameter $\Delta$ independent of the timings of Jack's packets and, at each point of the process, inserts her own packet into the channel between Jack and Steve. Bob collects and removes the packets inserted by Alice. 
	
	\textit{\textbf{Analysis}}: ({\em Covertness}) First, for each of the two possible locations of Willie, we show that communication is covert. Next, we calculate the number of covert packets transmitted by Alice.
	
	({\em Setting 1-Willie is between Alice and Bob, as show in Fig.~\ref{fig:SysMod2a}}): Willie observes the packets on the channel between Alice and Bob and decides whether Alice has inserted packets intended for Bob ($H_1$) or not ($H_0$). Note that $H_1$ and $H_1$ correspond to Poisson processes with rates $\lambda+\Delta$ and $\lambda$ respectively. By the Neyman-Pearson lemma~\cite[Ch. 3.2 and 13.1]{lehmann2005testing}, an optimal hypothesis test that minimizes the sum of error probabilities is the likelihood ratio (LRT) between the null hypothesis $H_0$ and the alternative hypothesis $H_1$ is given by~\cite[Ch. 3.5.2]{karr1991point}:
	\begin{align} 
	\label{eq:lrt}
	\Lambda(n) = \frac{\mathbb{P}_{N_1}(n)}{  \mathbb{P}_{N_0}(n)},
	\end{align}
	\noindent where $n$ is the number of packets that Willie observes in $[0,T]$, $\mathbb{P}_{N_0}(n) = \mathbb{P}(N_0=n)$ is the probability mass function (pmf) of the number of packets $N_0$ that Willie observes under the null hypothesis $H_0$ corresponding to a Poisson process with rate $\lambda$, and $\mathbb{P}_{N_1}(n)= \mathbb{P}(N_1=n)$ is the pmf for the number of packets $N_1$ that Willie observes under hypothesis $H_1$ corresponding to a Poisson process with rate $\lambda+\Delta$. 
	\noindent Suppose Alice sets 
	\begin{align} \label{eq:01}
	\Delta \leq \epsilon\sqrt{\frac{2\lambda}{T}}.
	\end{align}
	
	By~\eqref{eq:lrt}, we can see that the number of packets observed during the time interval of length $T$ is a sufficient statistic by which Willie can perform the optimal hypothesis test to decide whether Alice transmits or not. For any test on the number of packets during time $T$~\cite{bash_jsac2013},
	\begin{align} \label{eq:th10e0001} \mathbb{P}_{FA}+\mathbb{P}_{MD} \geq 1-     \sqrt{\frac{1}{2} \mathcal{D}(\mathbb{P}_{N_0} || \mathbb{P}_{N_1})} .
	\end{align}
	where $\mathcal{D}(\mathbb{P}_{N_0} || \mathbb{P}_{N_1})$ is the relative entropy between $\mathbb{P}_{N_0}$ and $\mathbb{P}_{N_1}$. 	Next, we show how Alice can lower bound the sum of average error probabilities by upper bounding $  \sqrt{\frac{1}{2} \mathcal{D}(\mathbb{P}_{N_0} || \mathbb{P}_{N_1})} $. For the given $ \mathbb{P}_{N_0}$ and $ \mathbb{P}_{N_1}$ the relative entropy is \cite{GrunerDKL}
	\begin{align}\nonumber
	\mathcal{D}(\mathbb{P}_{N_0} || \mathbb{P}_{N_1}) =  \Delta \cdot T - \lambda T \ln{\left( 1+\frac{\Delta}{\lambda}\right)}\leq \frac{\Delta^2}{2 \lambda}T \leq 2 \epsilon^2.
	\end{align}
	\noindent where the second to last step is true because $\ln(1+x) \geq x-\frac{x^2}{2}$ for  $x\geq 0$, and the last step is due to the definition of $\Delta$ given in~\eqref{eq:01}. Consequently, $\sqrt{\frac{1}{2} \mathcal{D}(\mathbb{P}_{N_0} || \mathbb{P}_{N_1})}  \leq \epsilon$, and thus$\mathbb{P}_{FA}+\mathbb{P}_{MD} \geq 1-\epsilon$ for $\Delta = \mathcal{O} \left(\sqrt{\lambda/ T}\right) $. 
	
	({\em Setting 2-Willie is between Bob and Steve, as show in Fig.~\ref{fig:SysMod2b}}): Willie observes the packets transmitted by Bob. Since Alice inserts her own packets independent of the channel, her insertion does not change the timing of Jack's packets. Since Bob removes Alice's inserted packets, Willie observes the original timings of the packets transmitted by Jack, and thus Alice and Bob's communication is covert. 
	
	({\em Number of Covert Packets}) Alice inserts packets according to a Poisson process with rate $\Delta$. Let $\tau(i)$ denote the time that Alice inserts the $i^{\mathrm{th}}$ packet, and $N_a$ denote the number of packets inserted by Alice. We focus on $\mathbb{P}\left(N_a \geq i\right)$. By~\eqref{def6e3}, 
	\begin{align}
	\nonumber \mathbb{P}\left(N_a \geq \epsilon \sqrt{\lambda T} \right) &= \mathbb{P}\left(\tau\left({\epsilon \sqrt{\lambda T}}\right) \leq T \right) = \mathbb{P}\left(\sum_{k=1}^{\epsilon \sqrt{\lambda T}} A_i \leq T \right),
	\end{align}
	\noindent where the $A_i$s are i.i.d. exponentially distributed IPDs with mean $\Delta^{-1}= \frac{1}{\epsilon} \sqrt{\frac{T}{2 \lambda }}$, which goes to infinity as $T\to \infty$. We introduce $A'_i = \sqrt{T} A_i$; $A'_1, A'_2, \ldots$ is a sequence of i.i.d. exponential random variables with finite mean $(\Delta \sqrt{T})^ {-1} = 1/{\epsilon \sqrt{2\lambda}}$ and variance ${1}/{2 \lambda \epsilon^2 }$. Consider 
	\begin{align}
\nonumber \lim_{T\to \infty}\mathbb{P}\left(N_a \geq \epsilon \sqrt{\lambda T} \right)=\lim_{T\to \infty} \mathbb{P}\left(\sum_{k=1}^{\epsilon \sqrt{\lambda T}} {A'_i} \leq \sqrt{T}\right) &=\lim_{T\to \infty} \mathbb{P}\left(\sum_{k=1}^{\epsilon \sqrt{\lambda T}} \frac{A'_i}{\epsilon \sqrt{\lambda T}} \leq \frac{\sqrt{T}}{\epsilon \sqrt{\lambda T}} \right),\\
		\label{eq:4} &= \lim_{T\to \infty} \mathbb{P}\left(\sum_{k=1}^{\epsilon \sqrt{\lambda T}} \frac{A'_i}{\epsilon \sqrt{\lambda T}} \leq \sqrt{2} \mathbb{E}[A'_i] \right)=1,
	\end{align}
	\noindent where the last step follows from the weak law of large numbers (WLLN) which yields $\sum_{k=1}^{\epsilon \sqrt{\lambda T}} \frac{A'_i}{\epsilon \sqrt{\lambda T}} \xrightarrow{P}\mathbb{E}[A'_i]$. By~\eqref{eq:4}, $\mathbb{P}\left(N_a \geq d\sqrt{\lambda T} \right) \to 1$, as $T \to \infty$, for all $d\leq \epsilon $. Consequently, Alice can insert $N_a = \mathcal{O} \left(\sqrt{\lambda T}\right)$ packets covertly.
	
	{\it (Converse)} To establish the converse, we provide an explicit detector for Willie that is sufficient to limit Alice's throughput across all potential transmission schemes (i.e., not necessarily insertion according to a Poisson process). Suppose that Willie observes a time interval of length $T$ and wishes to detect whether Alice transmits or not. Since he knows that the packet arrival process for the link between Jack and Steve is a Poisson process with parameter $\lambda$, he knows the expected number of packets in an interval $[0,T]$. Therefore, he counts the number of packets $S$ in this interval and performs a hypothesis test by setting a threshold $U$ and compares $S$ to $\lambda T+U$. If $S\leq\lambda T+U$, Willie decides $H_0$; otherwise, he decides $H_1$. Consider $\mathbb{P}_{FA}$,
	\begin{align}
\label{eq:PFAupperbound1}  \mathbb{P}_{FA} = \mathbb{P}\left(S>\lambda T + U \Big| H_0 \right) =\mathbb{P}\left(S-\lambda T > U \Big| H_0 \right)
	&\leq\mathbb{P}\left(|S-\lambda T| > U \Big| H_0 \right).
	\end{align}
	\noindent When $H_0$ is true, Willie observes a Poisson process with parameter $\lambda$; hence,
	\begin{align}
	\nonumber\mathbb{E}\left[S\Big|H_0\right]&=\lambda T,\\
	\nonumber\mathrm{Var}\left[S\Big|H_0\right]&={\lambda T}.
	\end{align}
	\noindent Therefore, applying Chebyshev's inequality on~\eqref{eq:PFAupperbound1} yields $\mathbb{P}_{FA} \leq \frac{\lambda T}{U^2 }$. Thus, $\forall 0<\alpha<1$, if Willie sets $U=\sqrt{\frac{\lambda T}{ \alpha}}$, he can achieve
	\begin{align}
	\nonumber \mathbb{P}_{FA} \leq \alpha.
	\end{align}
	Next, we will show that if Alice inserts $ \omega\left(\sqrt{\lambda T}\right)$ packets, she will be detected by Willie with high probability. Consider~$\mathbb{P}_{MD}$:
	\begin{align}
	\label{eq:1} \mathbb{P}_{MD}=\mathbb{P}\left(S\leq \lambda T + U \bigg| H_1\right) =\mathbb{P}\left(S\leq \lambda T + \sqrt{{\lambda T}/{ \alpha}} \bigg| H_1\right) &=\mathbb{P}\left(N_a+N_j\leq \lambda T + \sqrt{{\lambda T}/{ \alpha}} \right),\\
	\nonumber &=\mathbb{P}\left(N_j\leq \lambda T + \sqrt{{\lambda T}/{ \alpha}} -N_a\right),
	\end{align}
	\noindent where $N_a$ is the number of packets inserted by Alice and $N_j$ is the number of packets inserted by Jack. 
	We show in the Appendix~\ref{ap2t} that for all $\beta>0$,
	\begin{align}
	\label{eq:44}\lim\limits_{T \to \infty}\mathbb{P}_{MD}<\beta.
	\end{align}
	\noindent Since $\alpha$ and $\beta$ are arbitrary, $\mathbb{P}_{FA}+\mathbb{P}_{MD}$ is arbitrarily small whenever $N_a=\omega(\sqrt{\lambda T})$.
\end{proof}

\section{Renewal Channels (Assumption 2)} \label{sec:r1}
The packet arrival processes measured in many networks demonstrate non-Poisson behavior. Hence, in this section, we extend our results from Section III to the general renewal channel. Per Section II, we assume that the IPDs of Jack's transmitted stream are i.i.d. with pdf $p(x)$; thus, Jack's transmission rate is $\lambda=\left(\int_{0}^{\infty}x p(x)dx\right)^{-1}$. 

For Poisson channels, we took advantage of the fact that the superposition of two independent Poisson processes is a Poisson process. However, the superposition of two independent renewal processes is not necessarily a renewal process. Therefore, if Alice inserts her packets in the channel according to a renewal process, since the packet timings that Willie observes under $H_1$ ($\mathbb{P}_1$) is not a necessarily a renewal process, the derivation of $\mathbb{P}_1$ and the calculation of the relative entropy between $\mathbb{P}_1$ and  $\mathbb{P}_0$, which is required in the covertness analysis becomes challenging. Note that there is no special class of renewal processes (except Poisson processes) that makes the calculation easier; if the superposition of two ordinary renewal processes is an ordinary renewal process, then those processes are either Poisson~\cite{samuels1974characterization,ferreira2000pairs} or binomial-like processes~\cite{ferreira2000pairs}, which are not applicable to our scenarios. Therefore, we employ an alternative technique for Alice's insertion of packets. 

In~\cite{soltani2016allerton}, we employed the following technique: Alice and Bob employ  a two-phase scheme. In the first phase, Alice (slightly) slows down the packet stream to buffer packets. In the second phase, she generates a renewal process with a rate higher than Jack’s transmission rate. For each packet transmission during the second phase, Alice flips an unfair coin to decide whether to send one of her packets or one of Jack's packets. Although this technique is reasonable and its covertness analysis is accurate, the reliability analysis in~\cite{soltani2016allerton} relied on the approximation that a regular random walk can model Alice's buffer length in the second phase, which is not strictly true. Besides, it did not allow for the case where Willie is between Bob and Steve (Setting 2) in the covertness analysis. We can employ~\cite[Theorem 9.1]{iglehart1970multiple} which is also mentioned in~\cite[Theorem 4]{iglehart1972extreme} to relax the approximation in the reliability analysis. 

Here, we introduce another strategy that allows for accurate analysis. Alice and Bob employ a two-phase scheme. In the first phase, Bob transmits Jack's packets at a rate (slightly) smaller than Jack's packet rate $\lambda$ so as to build up a backlog of $N_b=\mathcal{O}(\sqrt{N})$ packets in his buffer. In this phase, Alice remains idle except for calculating $N_b$ by simulating Bob's buffering process. In the second phase, Alice replaces $N_b$ of Jack's packets with packets of her own and Bob replaces Alice's inserted packets with packets in his buffer. The second phase ends when the total number of (Alice's and Jack's) packets transmitted by Alice is $N$. 


In Lemma~\ref{lem:2}, we derive the number of packets that Bob can buffer when the total number of packets that Bob transmits is $N$. 
Consider $p^{-}(x,\rho_1)=(1-\rho_1) p\left(\left(1-\rho_1\right)x\right)$ which is the scaled version of $p(x)$, where $0<\rho_1<1$. Since $\int_{x=0}^{\infty}x p^{-}(x,\rho_1)dx = \frac{1}{1-\rho_1} \int_{x=0}^{\infty}x p(x) dx$, the renewal processes whose inter-arrival timings are governed by $p^{-}(x,\rho_1)$ has a smaller rate than that of $p(x)$. Lemma~\ref{lem:2} requires that $p^{-}(x,\rho_1)$ satisfies the following conditions ~\cite[Ch. 2.6]{kullback1968information} which are mentioned in~\cite[Theorem 1]{gurland1954regularity} as regularity conditions for maximum likelihood estimators with $f(x|\rho_1) = p^{-}(x,\rho_1)$:
\begin{align} 
\label{c1}\bullet  &\frac{\partial \log{p^{-}}}{ \partial \rho_1}, \frac{\partial^2 \log{p^{-}}}{\partial \rho_1^2},  \frac{\partial^3 \log{p^{-}}}{\partial \rho_1^3} \text{ exist, } \forall \rho_1\in (0,1)\\
\bullet\nonumber & \forall \rho_1 \in (0,1),  \bigg|\frac{\partial p^{-}}{\partial \rho_1}\bigg| < F(x), \text{ s.t. } \int_{x=0}^{\infty}F(x)dx<\infty, \\
&\nonumber \bigg|\frac{\partial^2 p^{-}}{\partial \rho_1^2}\bigg| < G(x),\text{ s.t. } \int_{x=0}^{\infty}G(x)dx<\infty \\
&\label{c2}  \bigg|\frac{\partial^3 \log p^{-}}{\partial \rho_1^3}\bigg| < H(x), \text{ s.t. } \int_{x=0}^{\infty}p(x)H(x)dx<\xi<\infty \text{ where }\xi \text{ is independent of }\rho_1\\
\label{c3}\bullet & \int_{x=0}^{\infty}\frac{\partial p^{-}(x,\rho_1)}{\partial \rho_1}\bigg|_{\rho_1=0}dx= \int_{x=0}^{\infty}\frac{\partial^2 p^{-}(x,\rho_1)}{\partial \rho_1^2}\bigg|_{\rho_1=0}dx=0
\end{align}
 Among the probability distributions that satisfy conditions (\ref{c1})-(\ref{c3}) are the generalized gamma distribution and its special cases: exponential distribution, Chi-squared distribution, Rayleigh distribution, Weibull distribution, Gamma distribution, and Erlang distribution. 
 
 We require that the support of $p(x)$ be $\mathbb{R}^{+}$ because 1) IPDs are positive; and 2) among the distributions with non-negative support, conditions (\ref{c1})-(\ref{c3}) do not satisfy for the distributions whose support is not $\mathbb{R}^{+}$, such as Pareto distribution, uniform distribution, and Beta distribution. Intuitively, the latter is required since Bob scales up the pdf of IPDs to $p^{-}(x,\rho)=(1-\rho_1)p(x (1-\rho_1))$ where $0<\rho_1<1$ is defined later. If the support of $p(x)$ is not $[0,\infty)$, then with high probability, the new pdf of the inter-packet delays $p^{-}(x,\rho)$ produces an IPD that does not fall in the support of $p(x)$. Hence, Willie will observe an inter-packet delay that cannot be generated from $p(x)$, and thus Willie detects Bob's buffering.
 

\begin{lem}\label{lem:2} Under the conditions given above for~\cite[Theorem 1]{gurland1954regularity} for the renewal process characterizing the packet timings on the link from Jack to Steve, Bob can covertly buffer $\mathcal{O}(\sqrt{N})$ packets while transmitting $N$ of Jack's packets, as long as $p^{-}(x,\rho_1)=(1-\rho_1) p\left(\left(1-\rho_1\right)x\right)$ satisfies conditions (\ref{c1})-(\ref{c3}). Conversely, if Bob buffers $\omega(\sqrt{N})$ packets while receiving $N$ of Jack's packets, there exists a detector that Willie can use to detect such a buffering with arbitrarily low sum of error probabilities $\mathbb{P}_{FA} + \mathbb{P}_{MD}$.
\end{lem}
\begin{proof}
	{\it (Achievability)} 
	
	\textit{\textbf{Construction}}: Since Alice does not insert any packets and she only relays Jack's packets, Bob receives Jack's packet stream. In this lemma, the term ``packet'' will refer to Jack's packets. For a fixed number of packets $N$, Bob scales up the IPDs by $\frac{1}{1-\rho_1}$ where $0<\rho_1<1$, i.e, if he receives the $i^{\mathrm{th}}$ packet at $\tau_p(i)$, he sends it at time $\frac{\tau_p(i)}{1-\rho_1}$, as shown in Fig.~\ref{fig:Streched}.
	
	First, we show that Bob can buffer $\mathcal{O}(\sqrt{N})$ packets, then we demonstrate covertness, and finally in the converse case we show that Bob cannot buffer $\omega(\sqrt{N})$ packets covertly.
	
	\textbf{\textit{Analysis}}: ({\em Number of Buffered Packets}) Bob sets 
	\begin{align}
	\nonumber \rho_1&= \frac{\epsilon}{\sqrt{c N}},
	\end{align}
	\noindent where $0<\epsilon<1$ and $c>0$ is a constant defined later. Note that the first phase ends at time $T_0\coloneqq\frac{\tau_p\left(N\right)}{1-\rho_1}$ when Bob transmits the $N^{\mathrm{th}}$ packet.	From $t=0$ to $t=T_0$, if Bob receives a packet of jack at time $\tau$, he transmits it at time $\frac{\tau}{1-\rho_1}$. Let $X_{p}(t)$ be the total number of packets received from Jack within the interval $[0,t]$, and  $\tau_{p}(i)$ be the time of arrival of the $i^{\mathrm{th}}$ packet from Jack. The total number of packets that Bob receives from Jack and the total number of packets that Bob buffers are $ X_p\left(  \frac{\tau_p\left(N\right)}{1-\rho_1}\right)$ and 
	\begin{align}
\nonumber 	
m = X_p\left(  \frac{\tau_p\left(N\right)}{1-\rho_1}\right) - N,
	\end{align}
	respectively.
	\begin{figure}
		
		\begin{center}
			\includegraphics[width=\textwidth,height=\textheight,keepaspectratio]{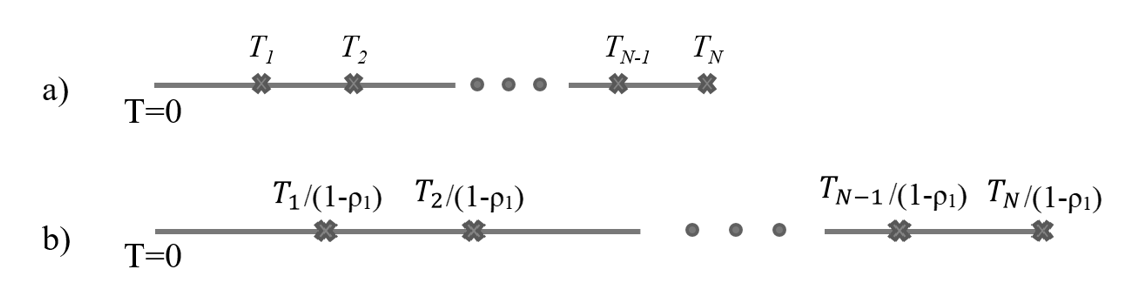}
		\end{center}
		\caption{a) Bob's received process b) The scaled version of Bob's received process when Bob uses a factor $\frac{1}{1-\rho_1}$.}
		\label{fig:Streched}
	\end{figure}
	\noindent We show in Appendices~\ref{ap1} and~\ref{ap2}, respectively, that
	\begin{align}
	\label{eq:5} \lim_{N \to \infty} \mathbb{P}\left( m \geq  \epsilon \sqrt{\frac{N}{ 4c}} \right) &= 1,\\
	\label{eq:51} \lim_{N \to \infty} \mathbb{P}\left( m \leq  \epsilon \sqrt{\frac{4 N}{ c}} \right) &= 1.
	\end{align}
	\noindent By~\eqref{eq:5},  for all $d\leq \epsilon \sqrt{\frac{1}{4c}}$, $ \mathbb{P}\left( m \geq  d \sqrt{N} \right) \to 1$, as $N \to \infty$. Therefore, Bob buffers $\mathcal{O}(\sqrt{N})$ packets when he transmits $N$ of Jack's packets. 
	
	({\em Covertness}) Now, we show that Bob's buffering is covert. If Willie is between Alice and Bob (Setting 1), he will not observe any changes in packet timings due to Bob's buffering, and thus the covertness follows immediately. Therefore, we present the analysis for the case where Willie is between Bob and Steve (Setting 2). We assume Willie knows the number of packets $N$ being slowed down and the scaling factor $1-\rho_1$ that Bob has possibly used. Upon observing the first $N$ packets, Willie decides whether Bob has not modified the packet timings ($H_0$), or he has slowed down those $N$ packets $(H_1)$. If Willie applies an optimal hypothesis test that minimizes $\mathbb{P}_{FA}+\mathbb{P}_{MD}$ on the IPDs, then arguments similar to those leading to~\eqref{eq:th10e0001} yield:
	\begin{align} \label{eq:th10e00011} \mathbb{P}_{FA}+\mathbb{P}_{MD} \geq 1-     \sqrt{\frac{1}{2}  \mathcal{D}(\mathbb{P}_0 || \mathbb{P}_1)}, \; 
	\end{align}
	\noindent where:
	\begin{align}
	\nonumber  \mathbb{P}_0&=\prod_{i=1}^{N} p(x_i),\\
	\nonumber \mathbb{P}_1&=\prod_{i=1}^{N} p^{-}(x_i,\rho_1).
	\end{align}
	\noindent Therefore, 
	\begin{align}
	\label{lem2:1}  \mathcal{D}(\mathbb{P}_0 || \mathbb{P}_1) = N   \mathcal{D}\left(p(x) || p^-(x,\rho_1)\right).
	\end{align}
	\noindent Since the regulatory conditions (\ref{c1}-\ref{c3}) hold,~\cite[Ch. 2.6]{kullback1968information} yields:
	\begin{align}
	\label{dklkhaf}\mathcal{D}\left({p}(x) || {p}^-(x,\rho_1)\right) =  \frac{c\rho_1^2}{2} +  \mathcal{O}\left(\rho_1^3\right) \text{ as } \rho_1 \to 0,
	\end{align}
	\noindent where $c$ is a positive constant derived in Appendix~\ref{apc4},
	\begin{align}
	\label{c4}  c=-1+ \int_{x=0}^{\infty}  p(x) {x^2} \left(\frac{d \log p(x)}{dx}\right)^2 dx.
	\end{align}
	\noindent Note that $c$ depends on $p(x)$. By~\eqref{lem2:1} and~\eqref{dklkhaf},
	\begin{align}
	\nonumber \mathcal{D}(\mathbb{P}_0 || \mathbb{P}_1) &= N \left(\frac{c \rho_1^2}{2} +  \mathcal{O}\left(\rho_1^3\right) \right) \text{ as } \rho_1 \to 0.
	\end{align}
	\noindent Because $\rho_1= \frac{\epsilon}{\sqrt{c N}}$, $\lim_{N\to\infty}\sqrt{\frac{1}{2} \mathcal{D}(\mathbb{P}_0 || \mathbb{P}_1)} =\lim_{N\to\infty} \epsilon \sqrt{\frac{N}{4 N}} < \epsilon$. Thus, by \eqref{eq:th10e00011}, $\mathbb{P}_{FA}+\mathbb{P}_{MD} {\geq}\, 1-\epsilon $ as $N\to \infty$ and Bob covertly buffers $\mathcal{O}(\sqrt{N})$ packets when he transmits $N$ of Jack's packets.

	{\it (Converse)} Since Willie knows $p(x)$, he knows the expected sum of the IPDs of $N$ packets. Therefore, he calculates the average observed IPD $S$ and performs a hypothesis test by setting a threshold $U$ and comparing $S$ with $\lambda^{-1} +U$. If $S\leq \lambda^{-1}+U$, he decides $H_0$; otherwise, he decides $H_1$. Observe
	\begin{align}
	\label{eq:PFAupperbound12} \mathbb{P}_{FA} &= \mathbb{P}\left(S>\lambda^{-1} + U \bigg| H_0 \right)=\mathbb{P}\left(S -\lambda^{-1} > U \bigg| H_0 \right) \leq\mathbb{P}\left(|S- \lambda^{-1}| > U \bigg| H_0 \right).
	\end{align}
	\noindent When $H_0$ is true, Willie observes a renewal process with rate $\lambda$, with variance $\sigma^2$; hence,
	\begin{align}
	\nonumber \mathbb{E}\left[S\Big|H_0\right]&= \lambda^{-1},\\
	\nonumber\mathrm{Var}\left[S\Big|H_0\right]&=\frac{ \sigma^2}{N}.
	\end{align}
	\noindent Therefore, applying Chebyshev's inequality on~\eqref{eq:PFAupperbound12} yields $\mathbb{P}_{FA} \leq \frac{\sigma^2 }{N U^2 }$. Therefore, if Willie sets $U=\sqrt{\frac{ \sigma^2}{  \alpha N}}$, for any $0<\alpha<1$, he achieves $\mathbb{P}_{FA}\leq \alpha$.
	
	Next, we will show that if Bob buffers $m=\omega(\sqrt{N})$ packets, he will be detected by Willie with high probability. 
	
	When Bob buffers packets, he will transmit $N$ packets during the time that Jack transmits $N+m$ packets. Therefore,
	$\tau_p(N+m) \leq S N < \tau_p(N+m+1)$.  Now, let us consider~$\mathbb{P}_{MD}$. When $H_1$ is true, $S N  \geq \tau_p(N+m)$. Thus:
	\begin{align}
	\nonumber \mathbb{P}_{MD}=\mathbb{P}\left(S\leq  \lambda^{-1}  + U \bigg| H_1\right) &=\mathbb{P}\left(S\leq  \lambda^{-1}  + U \bigg| {S }  \geq \frac{ \tau_p(N+m)}{N}\right),\\
	\nonumber &=\mathbb{P}\left( \frac{\tau_p(N+m)}{N}\leq S\leq  \lambda^{-1}  + U \bigg| {S }  \geq \frac{ \tau_p(N+m)}{N}\right),\\
	\nonumber &\leq\mathbb{P}\left(\frac{\tau_p(N+m)}{N} \leq  \lambda^{-1}  + U  \bigg| {S }  \geq \frac{ \tau_p(N+m)}{N}\right),\\
	\nonumber &=\mathbb{P}\left(\frac{\tau_p(N+m)}{N}\leq  \lambda^{-1}  + U \right),\\
	\nonumber &=\mathbb{P}\left( \frac{\tau_p(N+m)}{N+m}\leq \lambda^{-1} - \lambda^{-1} \frac{m+1}{N+m} + \frac{NU}{N+m} \right).
	\end{align}
	\noindent Note that $\tau_p(N+m)$ is the sum of $N+m$ i.i.d. random variables with mean $\lambda^{-1}$ and variance $\sigma^2$. Therefore, the central limit theorem (CLT) yields $\sqrt{N+m}\left(\frac{\tau_p(N+m)}{N+m}-\lambda^{-1}\right) \xrightarrow{D} Y$, where $Y\sim \mathcal{N}(0,\sigma^2)$ is a Gaussian random variable with mean zero and variance $\sigma^2$. Therefore, as $N \to \infty$, 
	\begin{align}
	\label{th2c3}
	 \mathbb{P}_{MD} \leq  \mathbb{P}\left( Y\leq - \lambda^{-1} \frac{m+1}{\sqrt{N+m}} + \frac{N U}{\sqrt{N+m}} \right)= \mathbb{P}\left( Y \leq  - \lambda^{-1} \frac{m+1}{\sqrt{N+m}} + {\sqrt{\frac{N}{N+m}}}\sqrt{\frac{ \sigma^2}{  \alpha }} \right),
	\end{align}
	\noindent where~\eqref{th2c3} is true since $U=\sqrt{\frac{ \sigma^2}{  \alpha N}}$. Thus, if $m=\omega(\sqrt{N})$ then $\lim\limits_{N \to \infty}\mathbb{P}_{MD}=0$. Combined with the results for the probability of false alarm above, if Bob collects $m=\omega\left({\sqrt{N}}\right)$ packets, Willie can choose a $U=\sqrt{\frac{\sigma^2}{N \alpha}}$ to achieve any (small) $\alpha>0$ and $\beta>0$ desired.
\end{proof}

Next, we leverage the results of Lemma~\ref{lem:2} to present and prove the results for packet insertion on a renewal channel. Although Alice and Bob do not know the actual location of Willie, their strategy guarantees covertness irrespective of Willie's location. 
Then, we conclude that if he analyzes the whole stream of packets transmitted by Alice, the communication is covert.

\begin{thm} \label{thWillieABS} In a renewal channel whose IPDs have $p(x)$, with conditions (\ref{c1}-\ref{c3}) true, Alice can covertly insert $\mathcal{O}(\sqrt{N})$ packets in a packet stream of length $N$. Conversely, if Alice attempts to insert $\omega(\sqrt{N})$ packets in a packet stream of length $N$, there exists a detector that Willie can use to detect her with arbitrarily low sum of error probabilities $\mathbb{P}_{FA} + \mathbb{P}_{MD}$.
\end{thm}

\begin{proof} 	
	{\it (Achievability)} 
	
	\textit{\textbf{Construction}}: Alice and Bob employ a two-phase scheme. During the buffering phase, Alice is idle but Bob slows down Jack's packets to build up packets in his buffer, i.e., if he receives a packet at time $\tau$, he transmits it at time $\frac{\tau}{1-\rho_2}$ where 
	\begin{align}
	\label{rho2Val} \rho_2= \epsilon \sqrt{\frac{1}{c N \psi}},
	\end{align}
	and $\psi$ is any constant that satisfies $0<\psi<1$. The first phase ends when Bob transmits the $\lfloor \psi N \rfloor^{\mathrm{th}} $ packet of Jack. From Lemma~\ref{lem:2}, Bob can buffer $N_b=\mathcal{O}\left(\sqrt{N }\right)$ packets covertly. Alice knows Bob's buffering process because she knows the timings of packets transmitted by Jack and $\rho_2$; thus, she calculates the number $N_b$ of packets buffered by Bob. In the second phase, Alice replaces $N_b$ of Jack's packets with packets of her own, and Bob replaces Alice's packets by Jack's packets in his buffer. Alice and Bob do this without changing the order of Jack's packets. Furthermore, Bob delays each packet in the second phase for $\phi$ seconds, where $\phi$ is the time elapsed between the moment that Bob receives the last packet in the first phase until the end of the first phase. We will later explain how this delay makes the pdf of Bob's first IPD in the second phase equal to $p(x)$.
	
	Since Willie cannot verify the source of the packets, Alice can choose any subset of size $N_b$ of packets transmitted by Jack in the second phase to replace them with her own packets. 
	Here, we propose a scheme where the locations of Alice packets are random. To decide whether to replace a packet, she uses a Bernoulli decision, i.e., each time she receives a packet from Jack, first she generates a random variable according to a Bernoulli distribution with $\mathbb{P}\left(\text{Success}\right)=\rho_3 = \frac{2 N_b}{N (1-\psi)}$. If she observes ``Success'', she replaces the packet; otherwise, she does not. She stops when she replaces the $N_b^{\mathrm{th}}$ packet. The second phase ends when the total number of (Alice's and Jack's) packets transmitted by Alice is $N$. At the end of the second phase, Alice will have $N_b$ of Jack's packets in her buffer. 
After the transmission, Alice and Bob will relay Jack's packets. Alice transmits Jack's oldest packet in her buffer and stores the newly received pack to keep the packets transmitted by Jack in order, and Bob, whose buffer is empty, forwards Jack's packets. 
	
		\textit{\textbf{Analysis}}: ({\em Covertness}) First, for each of the two possible locations of Willie, we show that communication is covert. Next, we calculate the number of covert packets transmitted by Alice.
	
	({\em Setting 1-Willie is between Alice and Bob, as show in Fig.~\ref{fig:SysMod2a} }):  Since Alice does not change packet timings and Willie is between Alice and Bob, Willie observes the original packet timings transmitted by Jack and covertness follows immediately.	
	
	({\em Setting 2-Willie is between Bob and Steve, as show in Fig.~\ref{fig:SysMod2b}}): Recall that Willie knows Alice and Bob's transmission scheme and parameters, the time they start and end each phase, and the scaling factor $1-\rho_2$ that Bob has used. We first assume Willie analyzes the packets in the two phases separately and show that the communication is covert. Then, we conclude that if he analyzes the whole stream of packets transmitted by Bob together, the communication is covert. 
	
	In the first phase, Bob slows down packets from Jack to buffer packets until he transmits packet $\lfloor \psi N \rfloor$ of Jack. By Lemma~\ref{lem:2}, Bob buffers $N_b=\mathcal{O}(\sqrt{N})$ packets, while for all $\epsilon$,
	\begin{align} 
	\nonumber \lim\limits_{N \to \infty}\mathcal{D}\left(\mathbb{P}_0^{(1)} || \mathbb{P}_1^{(1)}\right) &\leq 2 \epsilon^2,\\
	\nonumber \mathbb{P}_{FA}^{(1)}+\mathbb{P}_{MD}^{(1)} &\geq 1-\epsilon.
	\end{align}
	\noindent where $\mathbb{P}_0^{(1)}$ and $\mathbb{P}_1^{(1)}$ are joint pdfs of the IPDs in the first phase, when $H_0$ and $H_1$ are true respectively, and $\mathbb{P}_{FA}^{(1)}$ and $\mathbb{P}_{MD}^{(1)}$ are the probability of rejecting $H_0$ when it is true and the probability of rejecting $H_1$ when it is true, respectively in the first phase. Thus, Bob's buffering is covert. 
	\begin{figure}
	
		\begin{center}
			\includegraphics[width=16cm,height=10cm,keepaspectratio]{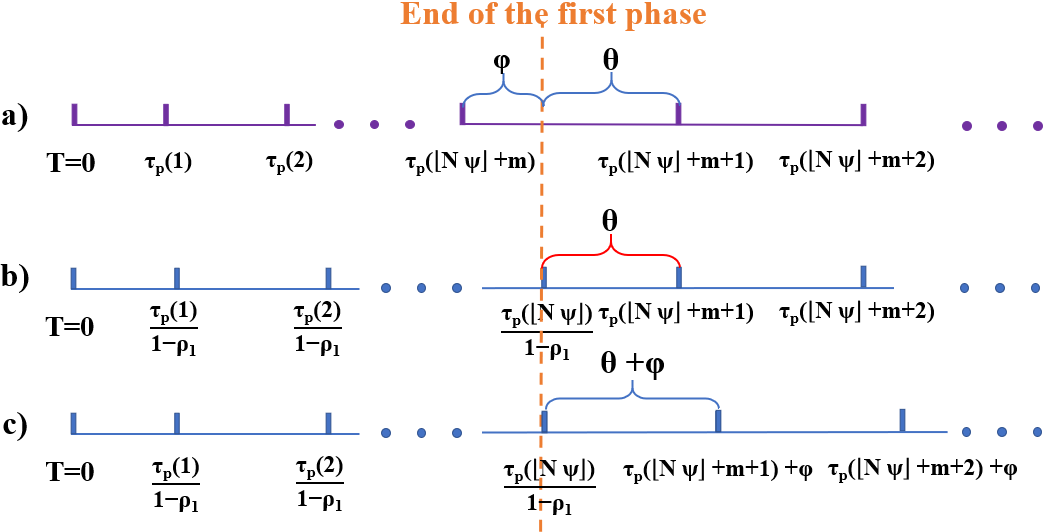}
		\end{center}
		\caption{a) Bob's packet arrival times b) Bob's packet departure times without delaying packets $\phi$ seconds in the second phase c) Bob's packet departure times with delaying packets $\phi$ seconds in the second phase. 
		}
		\label{fig:IPDs}
	\end{figure}

Next, we show that Willie observes Jack's original IPDs in the second phase. When the first phase ends Bob has buffered $m$ packets, transmitted $\lfloor N \psi \rfloor$ packets, and received $\lfloor N \psi \rfloor+m$ packets from Jack. Recall that in the second phase Bob delays each packet $\phi$ seconds, where $\phi$ is the time elapsed between the moment that Bob receives the last packet in the first phase ($t=\tau_p(\lfloor \psi N \rfloor + m)$) until the end of the first phase ($t=\frac{\tau_p{\lfloor N \psi \rfloor}}{1-\rho_1}$), i.e., $\phi=\frac{\tau_p{\lfloor N \psi \rfloor}}{1-\rho_1} -\tau_p(\lfloor \psi N \rfloor + m) $. Denote by $\theta$ 
the time elapsed between end of the first phase ($t=\frac{\tau_p{\lfloor N \psi \rfloor}}{1-\rho_1}$) and the moment that Bob receives the first packet in the second phase ($t=\tau_p(\lfloor \psi N \rfloor + m+1)$), i.e., 
$$\theta=\tau_p( \lfloor N \psi \rfloor+m + 1) - \frac{ \tau_p(\lfloor N \psi\rfloor )}{1-\rho_1}.$$ 
Since Bob delays the packets in the second phase $\phi$ seconds, Bob's first IPD in the second phase will be 
$$\phi+\theta = \tau_p( \lfloor N \psi \rfloor+m + 1) - \tau_p( \lfloor N \psi \rfloor+m),$$ 
which is Jack's original $\{\lfloor N \psi \rfloor+m + 1\}^{\mathrm{th}}$ IPD, and thus has the pdf $p(x)$ (see Fig~\ref{fig:IPDs}). Since all other IPDs in the second phase are also Jack's original IPDs, Willie observes the original IPDs transmitted by Jack and thus the covertness follows immediately for the second phase. Denote by $\mathbb{P}_0^{(2)}$ and $\mathbb{P}_1^{(2)}$ the joint pdfs of the IPDs in the second phase, when $H_0$ and $H_1$ are true, respectively, and by $\mathbb{P}_{FA}^{(2)}$ and $\mathbb{P}_{MD}^{(2)}$ the probability of rejecting $H_0$ when it is true and the probability of rejecting $H_1$ when it is true, respectively, in the second phase. Thus, 
	\begin{align}
	\nonumber\lim\limits_{N \to \infty}  {{\mathcal{D}(\mathbb{P}_0^{(2)} || \mathbb{P}_1^{(2)})}} &= 0,\\
	\mathbb{P}_{FA}^{(2)}+\mathbb{P}_{MD}^{(2)} &> 1-\epsilon.
	\end{align}
	\noindent Combined with the results of covertness for the first phase, if Willie analyzes the two sequences of packets in the first and second phase separately, the communication is covert, i.e., his sum of error probabilities in each phase is upper bounded by $1-\epsilon$ for all $\epsilon$. 
	
	Now assume that Willie analyzes the entire sequence of packets from the first and second phase together. Since $\mathbb{P}_0 = \mathbb{P}_0^{(1)} \mathbb{P}_0^{(2)}$ and $\mathbb{P}_1 = \mathbb{P}_1^{(1)} \mathbb{P}_1^{(2)}$, 
	\begin{align}
	\nonumber \mathcal{D}\left(\mathbb{P}_0||\mathbb{P}_1\right) = \mathcal{D}\left(\mathbb{P}_0^{(1)}||\mathbb{P}_1^{(1)}\right) + \mathcal{D}\left(\mathbb{P}_0^{(2)}||\mathbb{P}_1^{(2)}\right) \leq  \epsilon'^2.
	\end{align}
	\noindent Consequently, Alice can achieve $ \mathbb{P}_{FA}+\mathbb{P}_{MD} \geq 1-\epsilon'$ for all $\epsilon'$.
		
	({\em Number of Packets}) Recall that the first phase ends when Bob transmits the $\lfloor \psi N\rfloor^{\mathrm{th}}$ packet of Jack. Thus, replacing $N$ with $\lfloor \psi N \rfloor$ in~\eqref{eq:5} and ~\eqref{eq:51} yields $\mathbb{P}\left( N_b \geq  \epsilon \sqrt{\frac{\lfloor  \psi N\rfloor}{ 4c}} \right) \to 1$ and $\mathbb{P}\left( N_b \leq  \epsilon \sqrt{\frac{4 \lfloor \psi N \rfloor}{ c}} \right) \to 1$, respectively, as $N \to \infty$. Recall that $c>0$ is given in~\eqref{c4}. Since $ \mathbb{P}\left( N_b \geq  d \sqrt{\frac{4 \lfloor \psi N \rfloor}{ c}} \right) \to 1$ for all $d\leq  \epsilon$, Alice can insert $\mathcal{O}(\sqrt{N})$ packets in a packet stream of length $N$. 
	
	{\it (Converse)} The argument follows analogously to that of the converse in Lemma~\ref{lem:2}. Suppose that Willie observes $R=N+N_a$ packets and wishes to detect whether Alice has done nothing over the channel ($H_0$) or she has inserted $N_a$ packets. He calculates average observed IPD $S$ and sets a threshold $U$; if $S\geq \lambda^{-1}+U$, he decides $H_0$; otherwise, he decides $H_1$. We can show that if Willie sets $U=\sqrt{\frac{ \sigma^2}{ \alpha R}}$, for any $0<\alpha<1$, he achieves $\mathbb{P}_{FA}\leq \alpha$. Willie knows that if Alice chooses to insert packets, she will use the time of transmission of $N$ packets from Jack to do so. Therefore, if $H_1$ is true, then $S R  \leq \tau_p(N)$. Using this we can show that if Alice inserts $N_a=\omega(\sqrt{N})$ packets, then $\lim\limits_{N \to \infty}\mathbb{P}_{MD}=0$. Thus, if Alice inserts $N_a=\omega({\sqrt{N}})$ packets, Willie can choose a $U=\sqrt{\frac{\sigma^2}{R \alpha}}$ to achieve any (small) $\alpha>0$ and $\beta>0$ desired.
\end{proof}
\section{Discussion} \label{sec:dis1}

\subsection{Alice's insertion without the buffering phase}
In Theorem~\ref{thWillieABS}, Alice and Bob use a two-phase scheme. However, we could consider a simpler (one-phase) scheme. Alice generates a process with a (slightly) higher rate by generating packet transmission events for the following pdf $p^+(x,\rho_4)=\frac{p\left({x}/\left(1-\rho_4\right)\right)}{1-\rho_4}$, where 
\begin{align}
\label{rho4Val} \rho_4=  \frac{c'}{\sqrt{N}},
\end{align}
\noindent and $c'=\frac{\epsilon}{\sqrt{2c}}$. Note that $\rho_4<1$ for large enough $N$. Then:
\begin{enumerate}
	\item She buffers every packet she receives.
	\item Every time she generates a packet transmission event, she transmits one of Jack's packet from her buffer if one is there; if not, she sends a packet of her own.
\end{enumerate}
Although this scheme does not yield an infinite delay for packets unlike Alice and Bob's two-phase scheme in Theorem~\ref{thWillieABS} (see ~\eqref{eq:ap11}), it does not enable Alice to covertly insert $\mathcal{O}(\sqrt{N})$ packets; in fact, Alice cannot insert $f(N)=\omega(1)$ packets. 
\begin{thm} \label{thsimp} Consider the above scheme. There is no function $f(N)=\omega(1)$ such that $\lim\limits_{N \to \infty}\mathbb{P}\left(N_a \geq f(N) \right) =1$, where $N_a$ is the number of packets that Alice can insert packets intended for Bob in a packet stream of length $N$.
\end{thm}
See the proof in Appendix~\ref{ap10}.

\subsection{Packet delays due to buffering}
Our scheme requires Bob to slow down the packet stream to buffer packets, which results in packet delays. According to Lemma~\ref{lem:2}, since Bob receives the $i^{\mathrm{th}}$ packet at $\tau_p(i)$ and he sends it at time $\frac{\tau_p(i)}{1-\rho_1}$, Bob causes a delay of $\frac{\tau_p(i)\rho_1}{1-\rho_1}$ for the $i^{\mathrm{th}}$ packet. The average delay for the $N$ packets transmitted by Bob goes to $\infty$ as $N \to \infty$ because:
\begin{align}
\label{eq:ap11} 	\frac{\rho_1}{1-\rho_1}  \sum_{i=1}^{N} \frac{\mathbb{E}[\tau_p(i)]}{N}
= \frac{\rho_1}{1-\rho_1}  \sum_{i=1}^{N} i\frac{\lambda^{-1}}{N}= \frac{\rho_1}{1-\rho_1}\frac{N+1}{2}   \lambda^{-1}.
\end{align}
\noindent Note that each packet is delayed for an amount of time which is proportional to its time of arrival. According to the proof of Lemma~\ref{lem:2}, this large delay does not help Willie detect Bob's actions because Willie does not know the original packet timings but instead only knows the statistical properties of them, which change only slightly.



\subsection{Higher throughput via timing channel and bit insertion}
In this paper, Alice is allowed to buffer packets transmitted by Jack and release them when it is necessary; thus she is able to alter the timings of the packets. This suggests that Alice can also alter the timings of the packets to send information to Bob~\cite{verdubitsq} to achieve a higher throughput for sending covert information. However, this would require Alice and Bob to share a secret key (unknown to adversary Willie) prior to the communication which is not possible in many scenarios. Also, sending the information through IPDs (timing channel) is sensitive to the noise of the timings and thus not applicable in channels with a high level of noise in timings, such as complex channels in which multiple streams of packets are mixed and separated. In addition, a timing channel approach does not work over channels with zero capacity when packet timing is employed (e.g., deterministic queues). However, packet insertion works over such channels. Fig.~\ref{fig:complexchannel} depicts an example.

\begin{figure}
	\begin{center}
		\includegraphics[width=\textwidth,height=\textheight,keepaspectratio]{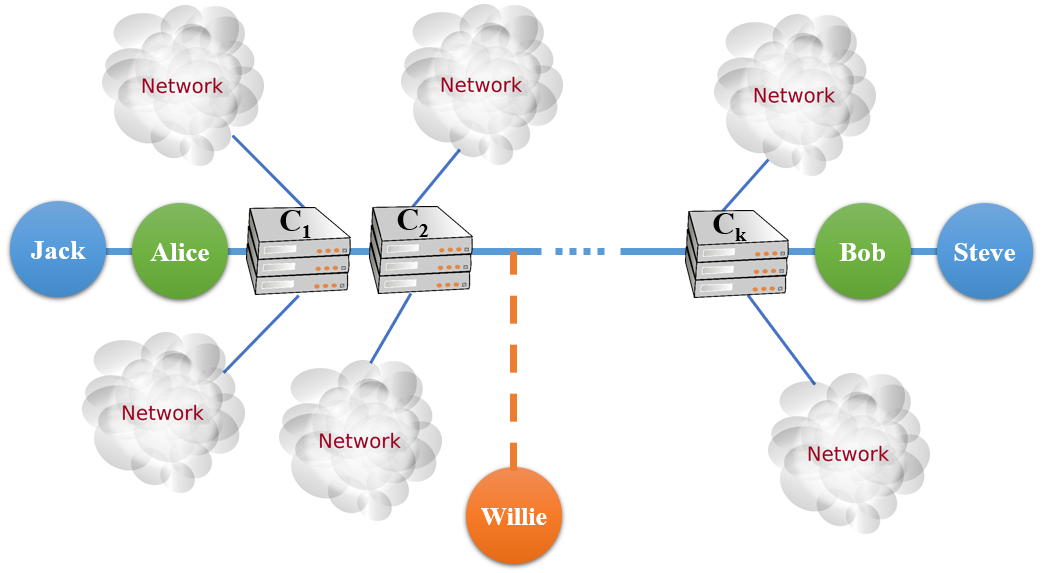}
	\end{center}
	\caption{Example of a complex network with high level of timing noise, in which multiple packet streams are mixed and separated in centers $C_1, C_2,\ldots, C_K$.}
	\label{fig:complexchannel}
\end{figure}

If we assume Alice and Bob can share a codebook and the altering of timings in the channel can be modeled by a queue, sending information via packet timing is studied for Poisson packet channels in~\cite[Theorem 2]{soltani2015covert} and for renewal channels in~\cite[Theorem 5]{soltani2016allerton}.

Another way to communicate covertly on a packet channel is bit insertion, where 
Alice inserts bits in a subset of the packets~\cite{soltani2018fundamentalbit}. This technique requires that packets have available space in their payload and a minimum of one bit in their header. In addition, Alice and Bob need to share a secret prior to the communication. These conditions can be satisfied only in some scenarios such as video streaming applications with variable bit rate codecs.


\subsection{Covertness of scaling up/down the IPDs}
In Lemma~\ref{lem:2}, we showed that when Bob scales down the IPDs such that their pdf becomes $p^{-}(x,\rho_1)=(1-\rho_1) p\left(\left(1-\rho_1\right)x\right)$, if conditions (\ref{c1}-\ref{c3}) hold, Bob's scaling is covert as long as $\rho_1=\mathcal{O}(1/\sqrt{N})$. Similarly, we can show that if he scales up IPDs such that their pdf becomes $p^{+}(x,\rho_1)=(1-\rho_1)^{-1} p\left(x/\left(1-\rho_1\right)\right)$, Bob's scaling is covert as long as $\rho_1=\mathcal{O}(1/\sqrt{N})$ and conditions (\ref{c1}-\ref{c3}) hold when $p^{-}(\cdot)$ is replaced with $p^{+}(\cdot)$.
\subsection{Packets in Alice's buffer after the second phase}
The construction of Theorems~\ref{thWillieABS} implies that Alice will have $\mathcal{O}(\sqrt{N})$ packets in her buffer at the end of the second phase. Assuming that Alice will be always on the link, having $\mathcal{O}(\sqrt{N})$ packets in her buffer does not cause any problems since after Alice and Bob's communication is done, they will only relay Jack's packets. Note that Alice can insert the packets in her buffer into the channel according to the timings of a Poisson process with a small rate. The covertness analysis of this scheme is challenging and requires calculating the relative entropy between a renewal process and its superposition with a Poisson process, which is relegated to future work.



\section{Future Work}
A key goal is to establish the fundamental limits of packet insertion in channels whose packet timings follow a general point process. We will let Alice insert packets on the channel according to a Poisson process with a small rate, independent of the channel. Then, we plan to employ the results of Girsanov's theorem to calculate the relative entropy between the point process governing the timings of the packets on the channel and the superposition of the point processes with a Poisson process. Another future work is analyzing covert throughout when the packet timings of the channel follow a Poisson process with a variable rate; in this case, we expect to be able to exploit Willie's difficulty in estimating the current packet rate under $H_0$ to allow for the insertion of $\mathcal{O}({N})$ packets.

\section{Conclusion} \label{sec:con1}
We present two scenarios for covert communication on a packet channel. In a Poisson channel where packet timings are governed by a Poisson process, Alice inserts her own packets into the channel but does not modify the timing of other packets. We established that Alice can covertly transmit $\mathcal{O}(\sqrt{\lambda T})$ packets to Bob in a time interval of length $T$; conversely, if Alice inserts $\omega(\sqrt{\lambda T})$ packets, she will be detected by Willie. In a renewal channel where the packet timings are governed by a general renewal process, we showed that Alice can covertly insert $\mathcal{O}(\sqrt{N})$ packets into the channel in a packet stream of length $N$. Conversely, if she inserts $\omega(\sqrt{N})$ packets, she will be detected by Willie with high probability. 

\appendices




\section{Proof of~\eqref{eq:44}} \label{ap2t}
Define $\mathcal{A}=\{N_j+N_a\leq \lambda T + \sqrt{{\lambda T}/{ \alpha}}\}$ and $\mathcal{B}=N_j < b\sqrt{\lambda T}+\lambda T$. Employing~\eqref{eq:1} and the law of total probability yields:
\begin{align} \label{eq:2}
\mathbb{P}_{MD}=\mathbb{P}\left(\mathcal{A}\right) =   \mathbb{P}(\mathcal{A}|\overline{\mathcal{B}}) \mathbb{P}(\overline{\mathcal{B}})  + \mathbb{P}(\mathcal{A}|\mathcal{B}) \mathbb{P}(\mathcal{B}) \leq   \mathbb{P}(\mathcal{A}|\overline{\mathcal{B}}) + \mathbb{P}(\mathcal{B}).
\end{align}
Consider the first term on the right hand side (RHS) of~\eqref{eq:2}. Substituting the events $\mathcal{A}$ and $\mathcal{B}$ yields:
\begin{align}
\nonumber \mathbb{P}(\mathcal{A}|\overline{\mathcal{B}})  &= \mathbb{P}\left(N_j+N_a\leq \lambda T + \sqrt{{\lambda T}/{ \alpha}}\bigg|N_j \geq  b\sqrt{\lambda T}+\lambda T\right),\\ 
\nonumber &\stackrel{(a)}{\leq} \mathbb{P}\left(b\sqrt{\lambda T}+\lambda T +N_a\leq \lambda T + \sqrt{{\lambda T}/{ \alpha}}\bigg|N_j \geq  b\sqrt{\lambda T}+\lambda T\right),\\
\label{eq:3} &= \mathbb{P}\left(N_a \leq  \sqrt{{\lambda T}/{ \alpha}} - b\sqrt{\lambda T} \bigg|N_j \geq  b\sqrt{\lambda T}+\lambda T\right)= \mathbb{P}(N_a\leq  \sqrt{{\lambda T}/{ \alpha}} - b\sqrt{\lambda T}).
\end{align}
where $(a)$ is true since the condition in the probability is $N_j\geq b\sqrt{\lambda T}+\lambda T$. If Alice inserts $N_a=\omega\left({\sqrt{\lambda T}}\right)$ packets, then~\eqref{eq:3} yields
$$\lim\limits_{T \to \infty} \mathbb{P}(\mathcal{A}|\overline{\mathcal{B}}) =0.$$
Consider the second term on the RHS of~\eqref{eq:2}. From \cite[p. 40]{cox1962renewal},
\begin{align}
\nonumber \lim\limits_{T \to \infty} \mathbb{P}\left(N_j<b\sqrt{\lambda T} + \lambda T\right)=\frac{1}{\sqrt{2 \pi}}\int_{-\infty}^{b} e^{-\frac{x^2}{2}}dx.
\end{align}
\noindent Since $b$ is arbitrary, we can choose $b$ small enough such that $ \lim\limits_{T \to \infty} \mathbb{P}\left(N_j<b\sqrt{\lambda T} + \lambda T\right)<\beta$. Thus, if $N_a=\omega\left(\sqrt{\lambda T}\right)$, then $\lim\limits_{T \to \infty}\mathbb{P}_{MD}< \beta$  for all $\beta>0$.

\section{Proof of~\eqref{eq:5}} \label{ap1}

Note that $X_p(\cdot)$ and $\tau_p(\cdot)$ correspond to a renewal process whose inter-arrival pdf is $p(x)$. Let $M=N(1+\rho_1/2)$. Since $\rho_1= \frac{\epsilon}{\sqrt{c N}}$ and $m = X_p\left(  \frac{\tau_p\left(N\right)}{1-\rho_1}\right) - N$,
\begin{align}
\nonumber \mathbb{P}\left(m \geq \epsilon \sqrt{\frac{N}{4 c}}\right) =\mathbb{P}\left(m \geq \frac{\rho_1 N}{{2}}\right)&= \mathbb{P}\left( X_p\left(  \frac{\tau_p\left(N\right)}{1-\rho_1}\right) - N\geq  \frac{\rho_1 N}{{2}}\right),\\
\nonumber &=  \mathbb{P}\left(X_p\left(  \frac{\tau_p\left(N\right)}{1-\rho_1}\right)  \geq  N\left(1+ \frac{\rho_1}{{2}}\right) \right)=  \mathbb{P}\left(X_p\left(  \frac{\tau_p\left(N\right)}{1-\rho_1}\right)  \geq  M \right),\\
\nonumber &\stackrel{(b)}{=}  \mathbb{P}\left(\left(1-\rho_1\right) \tau_p\left(M\right)\leq \tau_p\left(N\right) \right) \stackrel{(c)}{=}  \mathbb{P}\left(\left(1-\rho_1\right) \sum_{i=1}^{M} A_i^{(J)} \leq \sum_{i=1}^{N} A_i^{(J)}  \right),\\
 &\stackrel{(d)}{=}  \mathbb{P}\left(\left(1-\rho_1\right) \sum_{i=N+1}^{M} A_i^{(J)} \leq \rho_1 \sum_{i=1}^{N} A_i^{(J)}  \right),\\
\nonumber  &=  \mathbb{P}\left( \sum_{i=N+1}^{M} \frac{A_i^{(J)}}{M -N} \leq  \frac{\rho_1 N}{(M-N)(1-\rho_1)}  \sum_{i=1}^{N} \frac{A_i^{(J)}}{N}  \right),\\ 
\label{eq:24}  &\stackrel{(e)}{=}\mathbb{P}\left(\mathcal{C}_N \right) \geq   \mathbb{P}\left(\mathcal{C}_N | \mathcal{D}_N  \right)\mathbb{P}\left(\mathcal{D}_N  \right).
\end{align}

\noindent where $(b)$ follows from~\eqref{def6e3}, $A_1^{(J)}, A_2^{(J)}, \ldots$ are the IPDs of Jack's transmitted stream, $(c)$ is true since~\eqref{def6e1} is true, $(d)$ follows from removing the common summands, $(e)$ follows from defining $\mathcal{C}_N=\left\{ \sum_{i=N+1}^{M} \frac{A_i^{(J)}}{M -N} \leq  \frac{\rho_1 N}{(M-N)(1-\rho_1)}  \sum_{i=1}^{N} \frac{A_i^{(J)}}{N}  \right\}$,  and the last step follows from the law of total probability with $\mathcal{D}_N=\left\{\frac{\rho_1 N}{(M-N)(1-\rho_1)}\leq 4\right\}$. Since $M=N \left(1+\frac{\rho_1}{2}\right)$ and $\rho_1= \frac{\epsilon}{\sqrt{c N}}$,
\begin{align} \lim\limits_{N \to \infty} \mathbb{P} \left(\mathcal{D}_N\right) = \lim\limits_{N \to \infty} \mathbb{P} \left(\frac{\rho_1 N}{(M-N)(1-\rho_1)}\leq 4\right) \label{eq:22} &= \lim\limits_{N \to \infty} \mathbb{P} \left(\frac{2}{1-\rho_1}\leq 4\right) =1 .
\end{align}
\noindent By~\eqref{eq:24} and~\eqref{eq:22},
\begin{align}
\label{apeq31} \lim\limits_{N \to \infty} \mathbb{P}\left(m \geq \epsilon \sqrt{\frac{N}{ 4 c}}\right)  &\geq  \lim\limits_{N \to \infty} \mathbb{P}\left(\mathcal{C}_N | \mathcal{D}_N  \right) \geq  \lim\limits_{N \to \infty} \mathbb{P}\left( \sum_{i=N+1}^{M} \frac{A_i^{(J)}}{M -N} \leq  {4} \sum_{i=1}^{N} \frac{A_i^{(J)}}{N}  \right)=1.
\end{align}
\noindent where the last step is true since the WLLN yields $\sum_{i=N+1}^{M} \frac{A_i^{(J)}}{M -N} \xrightarrow{P} \lambda^{-1}$, $ -4 \sum_{i=1}^{N} \frac{A_i^{(J)}}{N} \xrightarrow{P} -4 \lambda^{-1}$, and thus $ \sum_{i=N+1}^{M} \frac{A_i^{(J)}}{M -N}-4 \sum_{i=1}^{N} \frac{A_i^{(J)}}{N} \xrightarrow{P} -3\lambda^{-1}$, as $N \to \infty$\footnote{If $Y_N \xrightarrow{P} Y$ and $Z_N \xrightarrow{P} Z$, then $Y_N + Z_N \xrightarrow{P} Y + Z$ for any sequences of random variables $Y_N$ and $Z_N$ (\cite[prob. 5 p. 262]{shiryaev1996probability}).}. Hence, the proof is complete.

\section{Proof of~\eqref{eq:51}} \label{ap2} The argument follows analogously to that of~\eqref{eq:5}. If $A_1^{(J)}, A_2^{(J)}, \ldots$ are the IPDs of Jack's transmitted stream and $ M'=N(1+2\rho_1)$, then:
\begin{align}
\nonumber \mathbb{P}\left(m \leq \epsilon \sqrt{\frac{4N}{  c}}\right)  &=  \mathbb{P}\left( \sum_{i=N+1}^{M'} \frac{A_i}{M' -N} \geq  \frac{\rho_1 N}{(M'-N)(1-\rho_1)}  \sum_{i=1}^{N} \frac{A_i}{N}  \right)=  \mathbb{P}\left( \sum_{i=N+1}^{M'} \frac{A_i^{(J)}}{M' -N} \geq  \frac{1}{2}  \sum_{i=1}^{N} \frac{A_i^{(J)}}{N}  \right),
\end{align}
\noindent where the last step is true since $\frac{\rho_1 N}{(M'-N)(1-\rho_1)} = \frac{1}{2 \left(1-\rho_1\right)}\geq \frac{1}{2}$. Then, similar to the arguments that leads to~\eqref{apeq31}, we can show that
\begin{align}
\nonumber \lim_{N \to \infty }\mathbb{P}\left(\sum_{i=N+1}^{M'} \frac{A_i^{(J)}}{M' -N}\geq \frac{1}{2} \sum_{i=1}^{N} \frac{A_i^{(J)}}{N} \right)=1.
\end{align} 
Thus,~\eqref{eq:51} is proved. 
\section{Proof of~\eqref{c4}} \label{apc4} 
From [Ch. 2.6]\cite{kullback1968information}, $c$ is the Fisher information which is given by
\begin{align}
\label{eq:ap4}c=\int_{x=0}^{\infty} p(x) \frac{1}{p(x)^2} \left(\frac{\partial p^{-}(x,\rho_1)}{\partial \rho_1}\bigg|_{\rho_1=0}\right)^2 dx.
\end{align}
\noindent Since $p^{-}(x,\rho_1)=(1-\rho_1)p(x(1-\rho_1))$,
\begin{align}
\label{eq:ap6} \frac{\partial p^{-}(x,\rho_1)}{\partial \rho_1} \bigg|_{\rho_1=0} &= \frac{\partial \left(\left(1-\rho_1\right) p\left(\left(1-\rho_1\right)x\right)\right)}{\partial \rho_1} \bigg|_{\rho_1=0}
 =-p(x) - x \frac{d p(x)}{dx}.
\end{align}
\noindent Therefore,~\eqref{eq:ap4} yields
\begin{align}
\label{eq:ap5} c&=\int_{x=0}^{\infty} p(x) + 2x \frac{d p(x)}{dx} +  \frac{x^2}{p(x)} \left(\frac{d p(x)}{dx}\right)^2 dx=1+ \int_{x=0}^{\infty} 2x \frac{d p(x)}{dx} +  \frac{x^2}{p(x)} \left(\frac{d p(x)}{dx}\right)^2 dx.
\end{align}
\noindent Taking integral of the both sides of~\eqref{eq:ap6} yields
\begin{align}
\nonumber \int_{x=0}^{\infty}\left(p(x)  + x \frac{d p(x)}{dx}\right)dx &= -\int_{x=0}^{\infty}\frac{\partial p^{-}(x,\rho_1)}{\partial \rho_1}\Big|_{\rho_1=0}dx =0,
\end{align}
\noindent where the last step is true because of the regulatory condition ~\eqref{c3}. Consequently,
\begin{align}
\label{eq:ap7} \int_{x=0}^{\infty}  x \frac{d p(x)}{dx} dx  = - \int_{x=0}^{\infty} p(x)  dx  =-1.
\end{align}
\noindent By~\eqref{eq:ap7},~\eqref{eq:ap5} yields
\begin{align}
\nonumber c&=-1+ \int_{x=0}^{\infty}  \frac{x^2}{p(x)} \left(\frac{d p(x)}{dx}\right)^2 dx=-1+ \int_{x=0}^{\infty}  p(x) {x^2} \left(\frac{d \log p(x)}{dx}\right)^2 dx.
\end{align}

\section{Proof of Theorem~\ref{thsimp}} \label{ap10} The total number of (Alice's and Jack's) packets that Alice transmits is $N_{aj}=N_a+N-m$, where $N_a$ is the number of Alice's packets inserted into the channel, $N$ is the number packets transmitted by Jack, and $m$ is the number of packets in Alice's buffer when Alice's scheme ends. If $f(N)=\omega(1)$, then:
	\begin{align} \label{eq:a00}
	\mathbb{P}\left(N_a\geq f(N)\right) =  \mathbb{P}\left(N_{aj} \geq f(N) + N -m \right) .
	\end{align}
	\noindent Let $X_{p^+}(t)$ be the total number of packets transmitted by Alice within the interval $[0,t]$, and  $\tau_{p^+}(i)$ be the time of arrival of the $i^{\mathrm{th}}$ packet transmitted by Alice. Note that $X_{p^+}(\cdot)$ and $\tau_{p^+}(\cdot)$ correspond to a renewal process whose inter-arrival pdf is 
	$p^+(x,\rho_4)=\frac{p\left({x}/\left(1-\rho_4\right)\right)}{1-\rho_4}$. Also, recall that $\tau_{p}(i)$ is the time of arrival of the $i^{\mathrm{th}}$ packet transmitted by Jack. Since Jack transmits $N$ packets in a time interval of length $\tau_{p}\left( N\right)$, and Alice uses this time to transmit $N_{aj}$ packets, $N_{aj}=X_{p^+}\left(\tau_{p}\left( N\right)\right)$. By~\eqref{eq:a00},
	\begin{align} 
	\nonumber \mathbb{P}\left(N_a\geq f(N)\right)=  \mathbb{P}\left(N_{aj} \geq f(N) + N -m \right) 
	\nonumber &=\mathbb{P}\left(X_{p^+}\left(\tau_{p}\left( N\right)\right)\geq f(N)+N-m\right),\\
	\label{eq:a1}&=\mathbb{P}\left(\tau_{p}\left( N\right)\leq \tau_{p^+}\left(f(N)+N-m\right)\right),
	\end{align} 
	\noindent where the last step is true since~\eqref{def6e3} is true. Note that $\tau_{p}\left( N\right)=\sum_{i=1}^{N}A_i^{(J)}$, where $A_i^{(J)}$s are samples of $p(x)$, and $\tau_{p^+}\left(f(N)+N-m\right)=\sum_{i=1}^{f(N)+N-m}A_i^{(A)}$, where $A_1^{(A)}$s are samples of $p_1(x,\rho_4)$. Let ${C}_i=\frac{{A}_i^{(A)}}{1-\rho_4}$ for $1\leq i\leq {N}$. Therefore, ${C}_1, {C}_2, \cdots$ are samples of  $p(x)$. By~\eqref{eq:a1}
	\begin{align}
 \mathbb{P}\left(N_{a} \geq f(N) \right)
	=\mathbb{P}\left(\sum_{i=1}^{N}A_i^{(J)}\geq\sum_{i=1}^{f(N)+N-m}A_i^{(A)} \right)
	 \label{eq:ap2} =\mathbb{P}\left(\sum_{i=1}^{N}A_i^{(A)}\geq\left(1-\rho_4\right)\sum_{i=1}^{f(N)+N-m}{C}_i \right)
	 =\mathbb{P}\left( \mathcal{H}_N \right),
	\end{align}
	\noindent where
	\begin{align} 
	\label{eq:ap3} \mathcal{H}_N &= \left\{ \frac{1}{N}\sum_{i=1}^{N}A_i^{(J)}\geq\frac{\left(1-\rho_4\right)\left(f(N)+N-m\right)}{N}  \sum_{i=1}^{f(N)+N-m}\frac{{C}_i}{f(N)+N-m} \right\}.
	\end{align}
	\noindent Let  $\mathcal{I}_N = \left\{\frac{1}{f(N)+N-m}\sum_{i=1}^{f(N)+N-m}{C}_i \geq \lambda^{-1}\right\}$.	The law of total probability yields 
	\begin{align}
	  \mathbb{P}\left(N_a\geq f(N)\right)=\mathbb{P}\left(\mathcal{H}_N \right)
	\label{eq:ap10}   &=  \mathbb{P}\left(\mathcal{H}_n|\mathcal{I}_n\right) \mathbb{P}\left(\mathcal{I}_n\right) +  \mathbb{P}\left(\mathcal{H}_n|\overline{ \mathcal{I}_N}\right) \mathbb{P}\left(\overline{\mathcal{I}_N}\right).
	\end{align}
	\noindent Consider $\mathbb{P}(\mathcal{I}_N)$. Since $N-m>0$ and $f(N)=\omega(1)$, $f(N)+N-m\to \infty$ as $N\to \infty$. Note that $E[C_i]=\lambda^{-1}$ and thus the CLT yields $\lim\limits_{N \to \infty}P\left(\mathcal{I}_n\right)=\frac{1}{2}$. By~\eqref{eq:ap10},
	\begin{align}
 \mathbb{P}\left(N_a\geq f(N)\right) = \frac{1}{2}   \mathbb{P}\left(\mathcal{H}_n|\mathcal{I}_n \right) + \frac{1}{2}  \mathbb{P}\left(\mathcal{H}_n|\overline{ \mathcal{I}_N}\right)	\label{eq:ap9}  & \leq \frac{1}{2}   \mathbb{P}\left(\mathcal{H}_n|\mathcal{I}_n \right) + \frac{1}{2} .
	\end{align}
	\noindent Consider $\mathbb{P}\left(\mathcal{H}_n|\mathcal{I}_n \right)$, 
	\begin{align}
	\nonumber \mathbb{P}(\mathcal{H}_N|\mathcal{I}_N) &\leq \mathbb{P}\left(\sum_{i=1}^{N}\frac{A_i^{(J)}}{N}\geq\frac{(N+f(N)-m)(1-\rho_4)}{N} \lambda^{-1}\Bigg|\mathcal{I}_N\right),\\
	\nonumber &= \mathbb{P}\left(\sum_{i=1}^{N}\frac{A_i^{(J)}}{N}\geq\frac{N+f(N)-m}{N}(1-\rho_4) \lambda^{-1}\right),
	\\
	\nonumber &=\mathbb{P}\left(\sum_{i=1}^{N } \frac{A_i^{(J)}-\lambda^{-1}}{\sqrt{N} \lambda } \geq  \frac{(f(N)-m)(1-\rho_4)}{\sqrt{N}}-\rho_4 \sqrt{N}  \right).
	\end{align}
	\noindent By~\eqref{rho4Val}, $\rho_4 \sqrt{N}= c'$. Therefore,
	\begin{align}
	\nonumber \mathbb{P}(\mathcal{H}_N|\mathcal{I}_N)=\mathbb{P}\left(\sum_{i=1}^{N }\frac{ A_i^{(J)}-\lambda^{-1}}{\sqrt{N} \lambda} \geq  \frac{f(N)-m}{\sqrt{N}}-c'\frac{f(N)-m}{N}-c' \right) .
	\end{align}
	\noindent By the CLT, $\lim_{N \to \infty }\mathbb{P}(\mathcal{H}_N|\mathcal{I}_N)=\kappa$ where $0\leq \kappa \leq 1$. If $0 \leq \kappa <1$,~\eqref{eq:ap9} yields   $\mathbb{P}\left(N_a\geq f(N)\right) <1$. The value $\kappa=1$ is achievable only if $f(N)-m=-\omega(\sqrt{N})$. Intuitively, $m=\mathcal{O}(\sqrt{N})$ and thus $\kappa=1$ requires $f(N)=\mathcal{O}(\sqrt{N})-\omega(\sqrt{N})<0$. Consequently, there is no $f(N)=\omega(1)$ such that $\lim\limits_{N \to \infty}\mathbb{P}\left(N_a \geq f(N) \right) =1$.
\QEDA


\bibliographystyle{ieeetr}

\end{document}